\documentclass[leqno,11pt]{article}
\usepackage[english]{babel}
\usepackage{amsmath}
\usepackage{amsthm}
\usepackage{amssymb}

\usepackage{graphicx}
\usepackage{epsfig}
\usepackage{eufrak}
\usepackage{mathrsfs}
\usepackage{enumerate}
\usepackage{lipsum}
\usepackage[dvips]{color}
\newtheorem{theorem}{Theorem}[section]
\newtheorem{definition}{Definition}[section]

\newtheorem{lemma}{Lemma}[section]
\newtheorem{proposition}{Proposition}[section]
\newtheorem{corollary}{Corollary}[section]
\newtheorem{remark}{Remark}[section]

\def\C{\mathbb C}

\def\R{\mathbb R}

\def\N{\mathbb N}
\def\Q{\mathbb Q}

\def\Z{\mathbb Z}

\def\({{\rm (}}
\def\){{\rm )}}

\title{\bf Gauss sums, superoscillations and the Talbot carpet}

\author{F. Colombo,\  I. Sabadini,\ D.C. Struppa,\ 
 A. Yger\footnote{corresponding author}}

\newcommand{\Addresses}{{
  \bigskip
  \footnotesize

Fabrizio Colombo, Irene Sabadini \textsc{
  Politecnico di
Milano, Dipartimento di Matematica, Via E. Bonardi, 9 20133 Milano,
Italy}\par\nopagebreak \textit{E-mail addresses}, F.~Colombo:\ \texttt{fabrizio.colombo@polimi.it}, I.~Sabadini:\  \texttt{irene.sabadini@polimi.it} 

  \medskip

  Daniele. C.~Struppa, \textsc{The Donald Bren Presidential Chair in Mathematics, Chapman University, Orange, CA 92866, USA}\par\nopagebreak
  \textit{E-mail address}, D.C.~Struppa:\ \texttt{struppa@chapman.edu}

  \medskip

  Alain Yger (corresponding author), 
  \textsc{IMB, Universit\'e de Bordeaux, 33405, Talence, France}\par\nopagebreak
  \textit{E-mail address}, A.~Yger:\ \texttt{yger@math.u-bordeaux.fr}

}}

\begin{document}
\numberwithin{equation}{section}
\maketitle

\begin{abstract}
 We consider the evolution, for a time-dependent Schr\"odinger equation, of the so called Dirac comb. We show how this evolution allows us to recover explicitly (indeed optically) the values of the quadratic generalized Gauss sums. Moreover we use the phenomenon of superoscillatory sequences to prove that such Gauss sums can be asymptotically recovered from the values of the spectrum of any sufficiently regular function compactly supported on $\R$. The fundamental tool we use is the so called Galilean transform that was introduced and studied in the context on non-linear time dependent Schr\"odinger equations. Furthermore, we utilize this tool to understand in detail the evolution of an exponential $e^{i\omega x}$ in the case of a Schr\"odinger equation with time-independent periodic potential.
\end{abstract}

\section{Introduction}\label{intro}

\noindent
An intriguing phenomenon, which was first discovered in \cite{aav} in the context of weak values in quantum mechanics, goes under the name of superoscillations, and is now well understood mathematically \cite{acsst5}. To describe such phenomenon in a few words, we begin by noticing that the sequence of entire functions in $z \in \C$
$$
F_N(z,\omega) :=
\sum\limits_{\nu=0}^N \binom{N}{\nu}
\Big(\frac{1 + \omega}{2}\Big)^{N-\nu}
 \Big(\frac{1 - \omega}{2}\Big)^{\nu}\,
 e^{i(1-2\nu/N)z}
$$
converges in $A_1(\C)= \{F\in H(\C)\,;\, |F(z)|\leq A e^{B |z|}
\ {\rm for\ some}\, A,B\geq 0\}$ towards $e^{i\omega z}$. As a consequence,
(see \cite[\textsection 5]{acsst5} or also \cite{acsst3, AOKI, cssy18}) this implies that if we denote by $\phi^s(t,x)$ the solution to
 the time-dependent Schr\"odinger equation
\begin{equation}\label{sect0-eq1}
\Big(i
\frac{\partial}{\partial t} +
\frac{\partial^2}{\partial x^2}\Big)(\phi)(t,x) =0,
\end{equation}
 with initial value $\phi(0,x)=e^{isx}$, then
 the sequence of distributions
$$
\sum\limits_{\nu=0}^N  \binom{N}{\nu}
\Big(\frac{1 + \omega}{2}\Big)^{N-\nu}
 \Big(\frac{1 - \omega}{2}\Big)^{\nu} \phi^{1-2\nu/N},\quad N\in \N^*,
$$
converges to the solution $\phi^\omega(t,x)$ in the space of Schwartz distributions $\mathscr D'(\R^+_t\times \R_x,\C)$ (for evident reasons related to the localization of the spectra with respect to the variable $x$, there cannot be of course convergence in the corresponding space of tempered distributions $\mathscr S'_x(\R^+_t\times \R_x,\C)$). Looking at the sequence of the restrictions to $\R$ of the functions $F_N(z,\omega)$, we have exactly the {\it superoscillation} phenomenon: each entry in such a sequence (therefore called {\it superoscillating}) has frequencies that lie in $[-1,1]\times \Q$ while the value of $\omega\in \R$ can be taken anywhere, and in particular outside
$[-1,1]$; note moreover that if $x$ denotes the real variable corresponding to the restriction of $z$ to $\R$, the convergence of $\{F_N(x,\omega)\}_{N\geq 1}$ towards $e^{i\omega x}$ is uniform on any compact subset of $\R$, hence holds in particular in $\mathscr D'(\R,\C)$
(but not of course in $\mathscr S'(\R,\C)$).
Such a superoscillating sequence evolves according to the time-dependent
Schr\"odinger equation \eqref{sect0-eq1} as an approximating sequence in
$\mathscr D'(\R^+_t\times \R_x,\C)$ for $\phi^\omega$,
which is the reason why, with respect to the parabolic equation $\eqref{sect0-eq1}$, the family
$$
\Big\{ \binom{N}{\nu}
\sum\limits_{\nu=0}^N
\Big(\frac{1 + \omega}{2}\Big)^{N-\nu}
 \Big(\frac{1 - \omega}{2}\Big)^{\nu}\, \phi^{1-2\tau/N}(t,x)\,;\, \omega \in \R\Big\} \subset \mathscr S'_x(\R^+_t \times \R_x,\C)
$$
realizes what is called a {\it supershift} in the sense of $\mathscr D'(\R^+_t \times \R_x\,,\C)$ for the family $\{\phi^\omega\,;\, \omega \in \R\}$
\cite{acsst5, cssy18}.
\vskip 2mm
\noindent

A case of special interest occurs when, given $M>0$, one consider as initial value for the Schr\"odinger euqtion the so called the Dirac comb
\begin{equation}\label{sect0-eq3}
u_M(x) =\frac{2\pi}{M} \sum\limits_{k\in \Z}
\delta \Big( x - \frac{2k\pi}{M}\Big) = \sum\limits_{k\in \Z} e^{iMkx}
\end{equation}
in which case the solution to the initial value problem can be easily computed as we show in Section $3$ (the Poisson summation formula provides the equality in \eqref{sect0-eq3}).
The reason why the Dirac comb is particularly interesting as initial datum for the time-dependent Schr\"odinger equation comes from optics. It is indeed well known that Fresnel diffraction in the so-called near-field Fresnel zone beyond the (vertical) reception plane $tOx$ (where the $t$-axis figures the horizontal direction orthogonal to the reception plane,
$t>0$ being interpreted to the distance to this plane) can be described from the mathematical point of view by the Fresnel integral
$$
\int_0^x e^{iy^2}\, dy.
$$
More precisely, if the variable $t$ denotes the distance to the
reception plane in the near-field Fresnel zone beyond the incidence and
$u_M(x)$ is the Dirac comb that models a periodic grating along the vertical axis $x'Ox$ in the reception plane, then what we observe beyond the shadow zone at a distance $t>0$ from the reception plane (when one remains in the near-field Fresnel zone) is the distribution
$\phi_{u_M}~:= \sum_{k\in \Z} e^{-i (Mk)^2 t} e^{iMkx}$. Since it is easily seen that this is indeed
a distribution of order $0$ with respect to the variable $t$, one can naturally restrict it to each line $\{(\tau,x)\,;\, x\in \R\}$ for any $\tau>0$. When
$\tau=t_{M,p,q}$ is exactly a rational fraction $((2\pi)/M^2)\times p/q$ of $2\pi/M^2$, where $p,q\in \N^*$,
$p\in \{0,...,q-1\}$ coprime with $q$, such a restriction is in fact
a positive measure equal to
\begin{equation}\label{sect0-eq4}
(\phi_{u_M} (t_{M,p,q},x))_{|t=t_{p,q}}
= \delta (t-t_{M,p,q})\otimes
\Big(\sum\limits_{\kappa=0}^{q-1} G(-p,\kappa,q)\,
\Big( \frac{2\pi}{Mq}
\sum\limits_{k\in \Z}
\delta \Big( x - \frac{2k\pi}{M} - \frac{2\pi \kappa}{Mq}\Big)\Big)\Big),
\end{equation}
where $G(-p,\kappa,q)$ is the {\it quadratic generalized Gauss sum}
defined as
\begin{equation}\label{sect0-eq5}
G(-p,\kappa,q) = G(-p,\kappa,q) = \sum\limits_{\ell=0}^{q-1} e^{2i\pi (-p \ell^2 +\kappa \ell)/q}.
\end{equation}
An intriguing discovery appears here : as we will recall it in \textsection \ref{section3}, quadratic generalized Gauss sums \eqref{sect0-eq5} are quantities of arithmetic nature, which an optical device (namely the diffraction through a periodic grating as pattern) allows to recover {\it optically}. In particular, the vanishing of $G(-p,\kappa,q)$ when $q=2q'$, $q'-\kappa \equiv 1$ modulo $2$ (see the incoming discussion in \textsection \ref{section3}, in particular \eqref{sect2-GS2}) is illustrated experimentally by the observation of the so-called {\it Talbot carpet} (\cite{berryHan80, berryGold88, Tayl03, dHV14}, se also \cite{WZX13} for an updated presentation and references).
\vskip 2mm
\noindent
The goal of this paper was therefore to connect these peculiar phenomena. Instead of recovering explicitly (optically) the quadratic generalized Gauss sums
through \eqref{sect0-eq4} (observe that the spectra of the Dirac combs involved are equal to full lattices in $\R$, thus involving arbitrary large frequencies), we will show
(in Theorem \ref{sect2-thm1}) that, given any sufficiently regular
($\mathscr C^2$ is sufficient) $\varphi$ on $\R_x$ with compact support in
$[-1,1]$ and such $\varphi(0)=1$, any quadratic generalized Gauss sum
$G(-p,\kappa,q)$ can be asymptotically recovered from the values of
the spectrum $\widehat \varphi$ of $\varphi$ on $[-\pi,\pi]$.
One should also refer to \cite{b4} and \cite{Gbur19} to point out the intimate connection between diffraction in the near-field Fresnel zone and superoscillations, as already observed in the pioneering works of G. Toraldo di Francia \cite{TorF52} in the fifties.
\vskip 2mm
\noindent
The main tool we introduce here was extensively used in
\cite{dHV14} in the context of the non-linear time dependent Schr\"odinger equation on the unit sphere, where it was called {\it Galilean transform}. The goal of \cite{dHV14} was to analyze the phenomenon leading to the Talbot carpet (previously mentioned) in a non-linear setting.
In the (much simpler) linear setting where we remain in this paper, such a Galilean transform consists in an elementary twisting operator between sets of solutions
of Cauchy problems (in $\mathscr D'([0,T[_t\times \R_x),\C)$, $T\in ]0,+\infty]$) for time-dependent Schr\"odinger equations of the form
\begin{equation}\label{sect0-eq6}
\Big(i
\frac{\partial}{\partial t} + \alpha
\frac{\partial^2}{\partial x^2}\Big)(\phi)(t,x) = V(t,x)\, \phi(t,x)
\end{equation}
where $V~: \R_t\times \R_x\rightarrow \R$ is a $\mathscr C^\infty$ real potential and $\alpha \in \R^*$. We will present in \textsection 2 such a tool and use it in \textsection 3 in the very particular case where $V\equiv 0$ and $\alpha=1$, in relation with the asymptotic computation of Gauss sums from the low-band spectrum of a smooth function with compact support. We will explain in \textsection 4 how such Galilean transform can be also exploited in order to compute the evolution in $\mathscr S'_x(\R^+_t \times \R_x)$ of $x\mapsto [e^{i\omega x}]$, $\omega \in \R$, under the time-independent Schr\"odinger equation
\begin{equation}\label{sect0-eq7}
\Big(i
\frac{\partial}{\partial t} + \alpha
\frac{\partial^2}{\partial x^2}\Big)(\phi)(t,x) = V(x)\, \phi(t,x)
\end{equation}
where $V: \R\rightarrow \R$ is a $\mathscr C^\infty$ periodic potential and $\alpha = \pm 1$. We formulate in Theorem \ref{sect4-thm2}
a closed formula for the evolution $\phi^{\alpha,\omega}$ of $[e^{i\omega x}]$ in
$\mathscr S'_x(\R^+_t\times \R_x,\C)$ according to the time-dependent Schr\"odinger equation \eqref{sect0-eq7}. We deduce from such a formulation in Theorem \ref{sect4-thm3} an approximated version of the fact that the family
$$
\Big\{
(t,x)\in \R^+_t \times \R_x \longmapsto \binom{N}{\nu}
\sum\limits_{\nu=0}^N
\Big(\frac{1 + \omega}{2}\Big)^{N-\nu}
 \Big(\frac{1 - \omega}{2}\Big)^{\nu}\, \phi^{\alpha, 1-2\tau/N}(t,x)\,;\, \omega \in \R\Big\}
$$
realizes a supershift for the family $\{\phi^{\alpha,\omega}\,;\,
\omega \in \R\}$ in the sense of distributions on $\R_t^+ \times \R_x$.

\section{Twisting solutions of time-dependent Schr\"odinger equations}\label{section2}

We begin by giving the formal definitions of the objects that we will be using throughout this paper, and by clarifying what we mean when we talk about solutions of the Schr\"odinger equation.

By $\mathscr S'_x(\R_t^+\times \R_x,\C)$ we denote the $\C$-vector space of complex valued distributions on $\R_t^+ \times \R_x$ (that is on
$\R_t\times \R_x$ with support in $\R^+_t \times \R_x$) which are {\it tempered} with respect to the variable $x$. It will be useful to interpret this notion as follows : if $\mathbb S^1=\{w\in \R^2\,;\|w\|=1\}\simeq \R/(2\pi \Z)$ and
$\boldsymbol \pi~: (t,w)\in \R^+ \times (\mathbb S^1 \setminus \{(0,1)\}) \leftrightarrow (t,\pi(w))=(t,x)\in \R$ denotes the stereographic projection, a complex valued distribution $\phi$ on $\R^+_t \times \R_x$ is tempered with respect to the variable $x$ if and only if the pushforward distribution $(\boldsymbol \pi^{-1})_*(\phi)$
on $\R^+ \times (\mathbb S^1\setminus \{(0,1)\})$ extends as a distribution on $\R^+ \times \mathbb S^1$.
\vskip 2mm
\noindent
Let $u\in \mathscr S'(\R_x,\C)$. It can be easily checked that
$$
u~: \varphi \in \mathscr S'(\R_x,\C)
\longmapsto \langle u(x),\varphi(x)\rangle = \Big\langle \widehat u(\xi)\,,\,
\xi \longmapsto \frac{1}{2\pi}
\int_\R \varphi(x) e^{i\xi x}\, dx\Big\rangle
$$
evolves in $\R^+_t \times \R_x$ (in a unique way) as a solution $(t,x)\mapsto \phi_u(t,x)$ (in $\R^+_t \times \R_x$)  of the time-dependent Schr\"odinger equation
\begin{equation}\label{sect1-eq0}
\Big(i
\frac{\partial}{\partial t} +
\frac{\partial^2}{\partial x^2}\Big)(\phi_u)(t,x) =0.
\end{equation}
This solution $\phi_u$ belongs to $\mathscr S'_x(\R^+_t\times \R_x,\C)$ and is given as follows.
For any $\theta\in \mathscr D(\R_t,\C)$ and any $\varphi\in
\mathscr S(\R_x,\C)$, then
\begin{equation}\label{sect1-eq0bis}
\big\langle \phi_u (t,x)\,,\,
\theta (t) \otimes \varphi(x)\big\rangle =
\Big\langle \widehat u(\xi)\,,\, \xi
\longmapsto \frac{1}{2\pi}
\int_{\R^2} \theta(t)\, \varphi(x)\, e^{-i x^2 t} e^{i\xi x}\, dt\, dx\Big\rangle.
\end{equation}
More precisely : the distribution $\phi_u$ which action on test functions is defined by \eqref{sect1-eq0} satisfies \eqref{sect1-eq0bis}
in the sense of distributions on $\R^+_t \times \R_x$ and one has for all
$\varphi \in \mathscr S(\R_x,\C)$ that
$$
\lim\limits_{\varepsilon \rightarrow 0^+}
\big\langle \phi_u (t,x)\,,\,
\frac{\rho (t/\varepsilon)}{\varepsilon} \otimes \varphi(x)\big\rangle = \langle u(x),\varphi(x)\rangle
$$
whenever $\rho \in \mathscr C^\infty(\R^+_t,\R)$ with
${\rm Supp} \, \rho \subset [0,1[$ and $\int_{[0,1[} \rho(t)\, dt =1$.
\vskip 2mm

Let $T \in ]0,+\infty]$, $(t,x)\in [0,T[ \times \R \longmapsto V(t,x)$ be a real $C^\infty$ function in $[0,T[\times \R$ and $\alpha$ be a non-zero real constant.
Let $\mathscr D_{V,\alpha}([0,T[_t\times \R_x,\C)$ be the subspace of
$$
\mathscr D'([0,T[_t\times \R,\C) = \{\phi \in
\mathscr D'(]-\infty,T[_t\times \R_x,\C)\,;\, {\rm Supp}\, \phi \subset [0,T[  \times \R\}
$$
which elements are the distributions
$\phi \in \mathscr D'([0,T[_t\times \R_x,\C)$ which satisfy in $[0,T[_t\times \R_x$ (in the sense of distributions in $[0,T[_t\times \R_x$) the time-dependent Schr\"odinger linear equation
\begin{equation}\label{sect1-eq1}
\Big(i\frac{\partial}{\partial t} + \alpha
\frac{\partial^2}{\partial x^2}\Big)(\phi)  = V(t,x)\phi(t,x).
\end{equation}
\vskip 1mm
\noindent

\begin{definition}\label{sect1-def1}
{\rm Let $V$, $\alpha$ as above and $u\in \mathscr D'(\R,\C)$. The distribution $u$ is said to {\rm evolve as $\phi\in \mathscr D_{V,\alpha}([0,T[_t\times \R_x,\C)$} if and only if
\begin{equation}\label{sect1-eq2}
\forall\, \varphi
\in \mathscr D(\R_x,\C),\
\lim\limits_{\varepsilon \rightarrow 0^+}
\big\langle \phi(t,x)\,,\, \rho_\varepsilon(t)\, \varphi(x)\big\rangle =
\langle u(x),\varphi(x)\rangle
\end{equation}
for any approximation $(\rho_\varepsilon(t))_{\varepsilon >0}
= (\rho(t/\varepsilon)/\varepsilon)_{\varepsilon >0}$
of $\delta_0$ in $\mathscr D'([0,T[_t,\R^+)$ ($\rho \in \mathscr D([0,T[_t,\R^+)$ with $\int_0^1 \theta(t) dt =1$).}
\end{definition}
\vskip 1mm
\noindent
The following immediate lemma reflects the effect of a change of scaling along the $x$-direction.

\begin{lemma} Let $\alpha = \pm \gamma^2$ with $\gamma>0$. Then
\begin{equation}\label{sect1-eq3}
\phi \in \mathscr D'_{V,\alpha}([0,T[_t \times \R_x,\C)
\longmapsto \big((t,x) \mapsto
\phi (t,\rho x)\big)
\end{equation}
realizes a continuous isomorphism between
$\mathscr D'_{V,\alpha}([0,T[_t\times \R_x,\C)$ and
$\mathscr D'_{V_\gamma,\pm 1}([0,T[_t\times \R_x,\C)$,
where $V_\gamma (t,x) = V(t,\gamma x)$. Moreover, if
$u\in \mathscr D'(\R_x,\C)$ evolves as $\phi
\in \mathscr D'_{V,\alpha}([0,T[_t\times \R_x,\C)$, then
$x\mapsto u (\rho x)$ evolves as $(t,x)\mapsto \phi(t,\gamma x)$ in
$\mathscr D'_{V_\gamma,\pm 1}([0,T[_t \times \R_x,\C)$.
\end{lemma}

\begin{proof} Let $\phi \in \mathscr D'_{V,\pm \gamma^2}([0,T[_t\times \R_x,\C)$. Then $\phi (t,x) = \phi(t,\gamma x)(t,x/\gamma)$. Substituting
in \eqref{sect1-eq1}, one gets immediately that $(t,x)\mapsto \phi(t,\gamma x)
\in \mathscr D'_{V_\gamma,\pm 1}([0,T[_t\times \R_x,\C)$.
\end{proof}

\vskip 1mm
\noindent
Given $\omega \in \R$, let $e^{-t\omega D_x}$ be the continuous invertible {\it shift-operator} from $\mathscr D'([0,T[_t\times \R_x,\C)$ into itself which acts as
\begin{equation}\label{sect1-eq4}
e^{-t\omega D_x}~:
\phi \in \mathscr D'([0,T[_t \times \R_x,\C) \longmapsto
e^{-t\omega D_x} (\phi)~: (t,x) \mapsto \phi (t,x-t\omega).
\end{equation}
Let also $\mathbb M_\omega$ be the continuous invertible {\it modulation-operator}
from $\mathscr D'([0,T[_t\times \R_x,\C)$ into itself which acts as
\begin{equation}\label{sect1-eq5}
\mathbb M_\omega~:
\phi \in \mathscr D'([0,T[_t \times \R_x,\C) \longmapsto
\mathbb M_\omega (\phi)~: (t,x) \mapsto  e^{i\omega x}
\phi (t,x).
\end{equation}

\begin{proposition}\label{sect1-prop1}
Let $\alpha = \pm 1$ and $\omega \in \R$. The operator
\begin{equation}\label{sect1-eq6}
\mathbb G_V^\omega = \Big(e^{-t\omega D_x}
\circ \mathbb M_\omega \circ e^{-t\omega D_x}\Big)_{|\mathscr D'_{V,\pm 1}([0,T[_t\times \R_x,\C)}
\end{equation}
realizes a continuous isomorphism between
$\mathscr D_{V,\alpha}([0,T[_t\times \R_x,\C)$ and $\mathscr D_{V^\omega,\alpha}([0,T[\times \R,\C)$, where
\begin{equation}\label{sect1-eq7}
V^\omega (t,x) = V(t, x-2 t\omega).
\end{equation}
One has $(\mathbb G_V^\omega)^{-1} = \mathbb G_{V^\omega}^{-\omega}$.
Moreover, if $u\in \mathscr D'(\R_x,\C)$ evolves as $\phi$ in
$\mathscr D'_{V,\alpha}([0,T[ \times \R,\C)$, then the modulated distribution
$e^{i\omega x}\, u\in \mathscr D'(\R_x,\C)$ evolves as $\mathbb G_\omega (\phi)$
in $\mathscr D'_{V^{\omega},\alpha}([0,T[ \times \R,\C)$.
\end{proposition}

\begin{remark}\label{sect1-rem1}
{\rm The introduction of the operator $\mathbb G_V^\omega$, together with its inverse, as twisting operators between $\mathscr D'_{V,\alpha}([0,T[_t \times \R,\C)$
and $\mathscr D'_{V^\omega,\alpha}([0,T[_t \times \R_x,\C)$,
is motivated by the essential role it plays in the mathematical formalisation of Talbot effect (see \textsection 3) in the linear as well as in the non-linear settings \cite{dHV14}.}
\end{remark}

\begin{proof}
Take $\alpha =-1$ (the proof being the same when $\alpha =-1$).
Let $\phi \in \mathscr D'_{V,1} ([0,T[_t \times \R,\C)$.
Let $\tilde \phi(t,x) = \phi(t,x-t\omega)$, that is
$\phi(t,x) = \tilde\phi (t,x+ t\omega)$. Then \eqref{sect1-eq1} implies that, in the sense of distributions on $[0,T[_t \times \R_x$,
\begin{equation}\label{sect3-eq8}
\Big( i \Big(\frac{\partial}{\partial t} -
\omega \frac{\partial}{\partial x}\Big) +
\frac{\partial^2}{\partial x^2}\Big) (\tilde \phi) (t,x)
= V(t,x- t\omega)\, \tilde \phi(t,x).
\end{equation}
Let now $\check \phi : (t,x) \mapsto e^{i\omega x}
\tilde \phi$, that is $\tilde \phi(t,x) = e^{-i\omega x}
\check \phi(t,x)$, so that, substituting in \eqref{sect3-eq8}, one gets
\begin{equation}\label{sect1-eq9}
\Big (i
\Big(\frac{\partial}{\partial t}
- \omega \, \frac{\partial}{\partial x}\Big)
+ \frac{\partial^2}{\partial x^2}\Big) (\check \phi) (t,x) =
V(t,x- t\omega)\, \check \phi(t,x).
\end{equation}
Let finally $\mathring \phi : (t,x) \mapsto \check \phi(t,x-\omega t)$, that is $\mathring \phi(t,x) = \phi (t,x+t \omega)$. Then \eqref{sect1-eq9} implies that
$$
\Big(
i\frac{\partial}{\partial t} +
\frac{\partial^2}{\partial x^2}
\Big)(\mathring \phi) (t,x+ t\omega) =
V(t,x-t\omega)\, \mathring \phi(t,x + t\omega),
$$
that is
$$
\Big(
i\frac{\partial}{\partial t} +
\frac{\partial^2}{\partial x^2}
\Big)(\mathring \phi) (t,x) = V(t, x-2t\omega)\, \phi(t,x)
\Longleftrightarrow \mathring \phi \in \mathscr D'_{1,V^\omega}([0,T[ \times \R,\C).
$$
The invertibility of $\mathbb G_V^\omega$, together with the inversion formula $(\mathbb G_V^\omega)^{-1} = \mathbb G_{V^\omega}^{-\omega}$, follows as an immediate consequence of these computations. The last assertion follows from the definition of the operator $\mathbb G^\omega_V$.
\end{proof}

\noindent
Proposition \ref{sect1-prop1} admits the following immediate consequence.

\begin{corollary}\label{sect1-cor1}
Let $V$, $\alpha, \omega$ as in Proposition \ref{sect1-prop1}. If the distribution $[1]$ evolves as $\phi$ in $\mathscr D'_{V^{-\omega},\alpha}([0,T[_t\times \R_x,\C)$, then the modulated distribution $[e^{i\omega x}]$ evolves as
$\phi^\omega = \mathbb G^\omega_{V^{-\omega}} (\phi)$ in the space $\mathscr D'_{V,\alpha}([0,T[_t\times \R_x,\C)$.
\end{corollary}

\begin{proof}
One just needs to apply Proposition \ref{sect1-prop1} with $V^{-\omega}$ instead of $V$.
\end{proof}

\section{The case $V=0$, Talbot carpet and Gauss sums}\label{section3}

Let $M$ be a strictly positive number.

\begin{proposition}
The $2\pi/M$-periodic Dirac comb
$$
x \in \R \mapsto u_M(x) = \frac{2\pi}{Mq} \sum\limits_{k\in \Z}
\delta \Big(x - \frac{2k\pi}{M}\Big)
$$
evolves as
\begin{equation}\label{sect2-eq1}
(t,x) \longmapsto \phi_{u_M}(t,x) = \sum\limits_{k\in \Z} e^{-i (Mk)^2 t}\, e^{iMkx}
\end{equation}
in $\mathscr D'_{0,1}(\R^+_t\times \R_x,\C)$.
\end{proposition}

\begin{proof} Our reasoning here was originally inspired from the methods introduced in \cite[\textsection 3]{dHV14}, except that, in order to make it shorter, we suggest from the beginning the closed expression \eqref{sect2-eq1} for the evolution of $u_M$ in $\mathscr D'_{0,1}(\R^+_t\times \R_x,\C)$.
Observe that the evolution $(t,x)\mapsto \phi_{u_M}(t,x)$ which is suggested in \eqref{sect2-eq1} is well defined as a distribution with order $0$ in $t$ and at most $2$ in $x$ since for any continuous function $\theta~: [0,T[_t\rightarrow \C$ with compact support ${\rm Supp}\, \theta \subset [0,t_\theta]$ and any
$\varphi \in \mathscr D(\R_x,\C)$,
\begin{eqnarray*}
\big|\big\langle \phi_M(t,x)\,,\,
\theta(t) \otimes \varphi(x)\rangle\big|
&=& \Big|\sum \limits_{k\in \Z} \widehat \varphi (-Mk)\, \int_{[0,T[}
\theta (t)\, e^{-i (Mk)^2 t}\, dt\Big| \\
&\leq & t_\theta \|\theta\|_\infty
\Big(
\int_{{\rm Supp} \, \xi}
\Big(|\varphi(y)| + \frac{\pi^2}{3 M^2}
|\xi''(y)|\Big)\, dy\Big).
\end{eqnarray*}
For any $\tau>0$, it is therefore possible to restrict $(t,x)\mapsto \phi_M(t,x)$ to
the horizontal line $\{\tau\} \times \R$, thus obtaining a distribution $x\mapsto \phi_M(\tau,x)$ (with order less or equal to $2$) in $\mathscr D'(\{\tau\}\times \R_x,\C)$. One can check immediately that $(t,x)\rightarrow \phi_M(t,x)$ satisfies in the sense of distributions on $\R^+_t \times \R_x$ the Schr\"odinger equation \eqref{sect1-eq1}, where $V\equiv 0$ and $\alpha=1$.
One can also restrict $(t,x)\rightarrow \phi_M(t,x)$ to the horizontal line
$\{0\} \times \R_x$ and obtain then
$$
\phi_{u_M}(0,x) = \sum\limits_{k\in \Z} e^{iMkx} = u_M(x),
$$
where the second equality is just Poisson summation formula. The assertion in the proposition follows then from the principle of unicity in the
Schr\"odinger Cauchy problem with $V=0$ and $\alpha =1$, namely that, given
$\psi\in \mathscr D'([0,T[\times \R,\C)$,
\begin{equation}\label{sect2-eq2}
\Big(\Big(i\frac{\partial}{\partial t} +
\frac{\partial^2}{\partial x^2}\Big)(\psi)  =
0,\quad \psi(0,x) = [0]\Big) \Longrightarrow \psi =0.
\end{equation}
\end{proof}

\begin{corollary}\label{sect2-cor1}
Let $\varphi\in \mathscr \mathcal C^2(\R,\C)$ with compact support. The regularized truncated Dirac comb
\begin{equation}\label{sect2-eq3}
x\mapsto (u_M *\varphi)(x) =
\big[\sum\limits_{k\in \Z} e^{ikM x}*\varphi\big] =
\Big[\sum\limits_{k\in \Z} e^{ikMx} \widehat\varphi(kM)\Big]
\end{equation}
evolves to $(t,x)\mapsto \phi_{u_M}(t,x)*\varphi(x)$ in $\mathscr D'_{0,1}(\R^+_t\times \R_x,\C)$, where $\phi_{u_M}$ is defined in \eqref{sect2-eq1}.
\end{corollary}

\begin{proof} The representation \eqref{sect2-eq2} for $x\mapsto (u_M*\varphi)(x)$ follows from Poisson summation formula. Since the convolution operation commutes with the action of differential operators, the regularized distribution
$(t,x) \mapsto \phi_{u_M}(t,x)*\varphi(x)$ satisfies in the sense of distributions in $\R^+_t \times \R_x$ the Schr\"odinger equation \eqref{sect1-eq1} with
$V=0$ and $\alpha=1$. One has also $(\phi_{u_M}(t,x)*\varphi(x))_{t=0} =
\delta(t)\otimes ((u_M*\varphi)(x))$ and the assertion of the corollary follows once again from the unicity principle such as formulated in
\eqref{sect1-eq2}.
\end{proof}

\begin{remark}\label{sect2-rem1}
{\rm Let, for $K\in \N^*$, $r_{M,K}\in \mathscr D'(\R_x,\C)$ be the distribution
$$
x\in \R \longmapsto \sum\limits_{|k|> K} e^{iMkx}\, \widehat \varphi (kM) =
\sum\limits_{|k|>K} (e^{iMkx}*\varphi) (x).
$$
As in Corollary \ref{sect2-cor1}, it evolves in $\mathscr D'_{0,1}
(\R^+_t\times \R_x,\C)$ towards
$$
(t,x) \in \R^+ \times \R \longmapsto \Big(
\sum\limits_{|k|>K} e^{-i (Mk)^2 t} e^{iMkx}\Big)* \varphi(x)  =
\Big[\sum\limits_{|k|>K} e^{-i (Mk)^2 t} e^{iMkx}\, \widehat \varphi(kM)\Big],
$$
with the uniform control with respect to $(t,x)\in \R^+\times \R$
\begin{equation}\label{sect2-eqrem1}
\Big|
\sum\limits_{|k|>K} e^{-i (Mk)^2 t} e^{iMkx}\, \hat\xi(kM)
\Big| \leq 2\, \frac{\|\varphi''\|_1}{K}.
\end{equation}
}
\end{remark}

\begin{proposition}\label{sect2-prop2}
Let $q\in \N^*$, $p\in \{0,...,q-1\}$ coprime with $q$ and $t_{M,p,q} = 2\pi/M^2 \times p/q$.
Then the following equality holds in $\mathscr D'(\{t_{p,q}\}
\times \R,\C)$ :
\begin{multline}\label{sect2-eq4}
\Big(\sum\limits_{k\in \Z}
e^{-i (Mk)^2 t}\, e^{iMk x}\Big)_{|\{t_{M,p,q}\}
\times \R} \\
= \delta (t-t_{M,p,q})
\otimes \Big(
\sum\limits_{\kappa =0}^{q-1}
G(-p,\kappa,q)\, \Big(
\frac{2\pi}{Mq}
\sum\limits_{k\in \Z}
\delta \Big( x - \frac{2k\pi}{M} - \frac{2\pi \kappa}{Mq}\Big)\Big)\Big),
\end{multline}
where $G(-p,\kappa,q)$ denotes the generalized quadratic Gauss sum
\begin{equation}\label{sect2-eq6}
G(-p,\kappa,q) = \sum\limits_{\ell=0}^{q-1} e^{2i\pi (-p \ell^2 +\kappa \ell)/q}.
\end{equation}
\end{proposition}

\begin{proof}
Since $t_{M,p,q} = 2\pi/M^2 \times p/q$, one has, after organizing
the indexes $k$ in the sum on the right-hand side of \eqref{sect2-eq1}
with respect to their classes in $\Z/q\Z$ and using once more the Poisson summation formula, that
\begin{eqnarray*}
(\phi_{u_M})_{|t=t_{M,p,q}}
&=& \delta (t-t_{M,p,q})
\otimes \Big(
\sum\limits_{\ell=0}^{q-1}
\sum\limits_{k'\in \Z} e^{2i\pi ((k'q+\ell)^2 p/q) + i M(q k'+\ell) x}\Big) \\
&=& \delta (t-t_{M,p,q}) \otimes \Big(
\sum\limits_{\ell=0}^{q-1}
e^{-2i\pi \ell^2 p/q + i M\ell x}
\sum\limits_{k'\in \Z} e^{i M q k' x}\Big) \\
&=& \frac{2\pi}{Mq}\, \delta (t-t_{M,p,q}) \otimes \Big(
\sum\limits_{\ell=0}^{q-1}
e^{-2i\pi \ell^2 p/q + i M\ell x}
\Big(\sum\limits_{k'\in \Z}
\delta \Big( x - \frac{2\pi k'}{qM}\Big)\Big)\Big)\\
&=& \frac{2\pi}{Mq}\, \delta (t-t_{M,p,q}) \otimes \Big(
\sum\limits_{\ell=0}^{q-1}
e^{-2i\pi \ell^2 p/q + 2i\pi k'\ell /q}
\Big(\sum\limits_{k'\in \Z}
\delta \Big( x - \frac{2\pi k'}{qM}\Big)\Big)\Big) \\
&=& \frac{2\pi}{Mq}  \delta (t-t_{M,p,q}) \otimes \Big(
\sum\limits_{k''\in \Z}
\Big(\sum\limits_{\ell=0}^{q-1} \sum_{\kappa =0}^{q-1}
e^{-2i\pi \ell^2 p/q + 2i\pi \kappa \ell /q}\, \delta \Big( x - \frac{2\pi k''}{M} - \frac{2\pi \kappa}{Mq}\Big)\Big)\Big)\Big) \\
&=& \frac{2\pi}{Mq}  \delta (t-t_{M,p,q}) \otimes
\Big(\sum\limits_{\kappa =0}^{q-1} G(-p,\kappa,q)
\, \sum\limits_{k\in \Z}\delta \Big( x - \frac{2\pi k}{M} - \frac{2\pi \kappa}{Mq}\Big)\Big)
\end{eqnarray*}
as expected.
\end{proof}

\begin{corollary} Let $M, p,q, t_{M,p,q}$ as in Proposition
\ref{sect2-prop2}. Let also $\varphi \in \mathscr C^2(\R,\C)$ with
$\varphi(0)=1$ and such that ${\rm Supp}\, \varphi \subset [-1,1]$.
Then
\begin{equation}\label{sect2-eq7}
G(-p,\kappa,q) = \frac{M q}{2\pi}\, \Big\langle \phi_{u_M}(t_{M,p,q},x)\,,\,
\varphi\Big(\frac{Mq}{2\pi}\Big( x - \frac{2\pi \kappa}{Mq}\Big)\Big)
\Big\rangle\quad (\kappa =0,...,q-1).
\end{equation}
\end{corollary}

\begin{proof}
The intersection of the support of the test function
$$
x\longmapsto \varphi\Big(\frac{Mq}{2\pi}\Big( x - \frac{2\pi \kappa}{Mq}\Big)\Big)
$$
(contained in
$2\pi \kappa/(Mq) + [-2\pi/(Mq),2\pi/(Mq)]$ ) with the support of $x\mapsto \phi_{u_M}(t_{M,p,q},x)$
consists in the single point $2\pi \kappa/(Mq)$~; the mass of the atomic measure $x\mapsto \phi_{u_M}(t_{M,p,q},x)$ at this precise point equals $2\pi\, G(-p,\kappa,q)/(Mq)$ by Proposition \ref{sect2-prop2}. Since $\varphi(0)=1$, the assertion \eqref{sect2-eq4} then follows.
\end{proof}

\vskip 2mm
\noindent
\begin{remark}\label{sect2-rem2}
{\rm If one specifies the choice of $M$ as $M=2\pi$ (which is always possible thanks to a rescaling of the $x$ axis),
in which case $u_M$ is the $1$-periodic Dirac comb $u(x) = \sum_{k\in \Z} \delta (x-k)$,
then \eqref{sect2-eq6} becomes
\begin{equation}\label{sect2-eq8}
G(-p,\kappa,q) = q\, \big\langle \phi(t_{p/q},x)\,,\,
\varphi(qx-\kappa)
\big\rangle = q\, (\phi(t_{p/q},x)*\varphi_{q,\kappa}(x))(0)
\quad (\kappa =0,...,q-1),
\end{equation}
where $\phi:=\phi_{2\pi}$ and $t_{p/q} := (2\pi)^{-1} p/q$ and
$\varphi_{q,\kappa}(x) =\varphi(-qx-\kappa)$.
}
\end{remark}

\noindent
{\it Generalized quadratic Gauss sums} $G(-p,\kappa,q)$, when $q\in \N^*$ and $p,\kappa \in \{0,...,q-1\}$ with $p$ coprime with $q$, are indeed quantities of arithmetic nature. They are deduced from the so-called {\it normal quadratic Gauss sums}
$G(a,0,c)$ ($a,c\in \N^*$ coprime) computed by Gauss thanks to the multiplicative formula
$$
G(a,b,cd) = G(ac,b,d)\, G(ad,b,c)
$$
provided that $c$ and $d$ are coprime. We recall that, for $a,c\in \N^*$ coprime, the normal quadratic Gauss sums are given by
$$
G(a,0,c) = \sum\limits_{\ell=0}^{c-1} e^{2i\pi a\ell^2/c} =
\begin{cases} \sqrt c \Big(\!\!\!\Big(\frac{c}{a}\Big)\!\!\!\Big) (1 +
e^{ia\pi/2})\quad {\rm if}\ c\equiv 0\ {\rm mod.}\ 4 \\
\sqrt c\, \Big(\!\!\!\Big(\frac{a}{c}\Big)\!\!\!\Big)\quad \; {\rm if}\ c\equiv 1\ {\rm mod.}\ 4 \\
0 \quad\quad \quad \quad \  {\rm if}\ c\equiv 0\ {\rm mod.}\ 4 \\
\sqrt c\ \Big(\!\!\!\Big(\frac{a}{c}\Big)\!\!\!\Big)\, i \ \ {\rm if}\ c\equiv 3\ {\rm mod.}\ 4,
\end{cases}
$$
where $\big(\!\!\!(p/q\big)\!\!\!\big)$ denotes the Jacobi symbol ($1$ if $p$ is a square modulo $q$ and $-1$ otherwise).
As for the values of the generalized Gauss sums $G(-p,\kappa,q)$, a summary is proposed
in \cite[Appendix A]{dHV14}, see also \cite{BEW98}. Here are a few partial results which illustrate the arithmetic nature of such generalized quadratic Gauss sums.
\begin{enumerate}
\item When $q$ is even and $u\in \{0,...,q-1\}$ is such that
$pu\equiv 1$ modulo $q$, then
\begin{eqnarray}\label{sect2-GS1}
G(-p,\kappa,q) &=&  e^{2i\pi \kappa^2 u^2/q}
\sum\limits_{\ell=0}^{q-1} e^{-2i\pi p (\ell -\kappa u)^2/q} = e^{2i\pi \kappa^2 u^2/q}
\sum\limits_{\ell' =0}^{q-1} e^{-2i\pi (\ell')^2/q} \nonumber \\
&=& e^{2i\pi \kappa^2 u^2/q}
\, \Big(\!\!\!\Big(\frac{p}{q}\Big)\!\!\!\Big)\times  \sqrt q \times
\, \begin{cases} 1 \ {\rm if}\ q\equiv 1\ {\rm mod.\ 4} \\
-i \ {\rm if}\ q \equiv 3\ {\rm mod.\ 4}.
\end{cases}
\end{eqnarray}
\item Suppose $q=2q'$, in which case $p$ is even. According to \cite{berryHan80} or also \cite[\textsection 3]{Tayl03}, one has
\begin{equation}\label{sect2-GS2}
\begin{split}
& q'-\kappa \equiv 1 \ {\rm mod.\ 2} \Longrightarrow G(-p,\kappa,q)=0 \\
& q'-\kappa \equiv 0\ {\rm mod.\ 2} \Longrightarrow G(-p,\kappa,q) =
\sqrt{2q}\, e^{i\theta_{p,q}(\kappa)}
\end{split}
\end{equation}
\end{enumerate}
\vskip 2mm
\noindent
As we already mentioned in the introduction
The vanishing of some generalized quadratic Gauss forms (see
\eqref{sect2-GS2}) is the mathematical justification for so called Talbot effect in optics \cite{berryHan80, berryGold88, Tayl03, dHV14}.
\vskip 2mm
\noindent
As we have just seen, quadratic generalized Gauss sums $G(-p,\kappa,q)$ are true arithmetic objects that surprizingly can be explicitly computed exactly through an optical device, namely the diffraction pattern, or Talbot carpet, generated by a periodic grating. In such an optical device, the initial datum $u_M$ is from the mathematical point of view a Dirac comb $x\mapsto \sum_{k\in \Z}
e^{iMkx}$ with arbitrary large frequencies. A natural question then arises : given an arbitrary sufficiently regular function $\varphi$ with support localized in $[-1,1]$ such $\varphi (0)=1$, can one recover asymptotically the quadratic generalized Gauss sums $G(-p,\kappa, q)$ from the values of the spectrum of $\varphi$ on the companion frequency domain $[-\pi,\pi]$~?
The mathematical justification for the phenomenom of superoscillations
which is present in optics \cite{TorF52,berryHan80,Gbur19} ensures in fact that the answer is yes.
This is the goal of our next Theorem \ref{sect2-thm1}.
\vskip 2mm
\noindent
Prior to state the result, let us introduce some notations. Given
$N,N'\in \N^*$, $\nu \in \{0,...,N-1\}$, $\nu'\in \{0,...,N'-1\}$ and
$\kappa \in \N$, let
\begin{equation}\label{sect2-eq9}
\omega_{\nu,\nu'}^{N,N'} (\kappa) :=
\exp \Big( - 2i\pi
\Big(\frac{1}{2} - \frac{\nu}{N}\Big)\Big(\kappa + \Big(\frac{1}{2} - \frac{\nu'}{N'}\Big)\Big)\Big).
\end{equation}

\begin{theorem}\label{sect2-thm1}
Let $(N_K)_{K\geq 1}$ and $(N'_K)_{K\geq 1}$ be two sequences of strictly positive integers such that
\begin{equation}\label{sect2-eq10}
\lim_{K\rightarrow +\infty} \frac{\log N_K}{K} =
\lim\limits_{K\rightarrow +\infty} \frac{\log N'_K}{K} = +\infty.
\end{equation}
Then, for any $q\in \N^*$, for any $\kappa \in \{0,...,q-1\}$,
for any $p\in \{1,...,q-1\}$ coprime with $q$,
for any $\varphi\in \mathcal C^2(\R,\C)$ with compact support in
$[-1,1]$ such that $\varphi(0)=1$,
\begin{multline}\label{sect2-eq11}
G(-p,\kappa,q) =
\lim\limits_{K\rightarrow +\infty} \\
\sum\limits_{k=-K}^K \sum\limits_{\nu=0}^{N_K}
\sum\limits_{\nu'=0}^{N_K'}
\Big(\frac{1}{2}+ \frac{k}{q}\Big)^{N-\nu}
\Big(\frac{1}{2}- \frac{k}{q}\Big)^{\nu}
\Big(\frac{1}{2} -kp\Big)^{N'-\nu'}
\Big(\frac{1}{2} + kp\Big)^{\nu'}
\omega_{\nu,\nu'}^{N,N'}(\kappa)\, \widehat \varphi\Big(
2\pi \Big(\frac{1}{2} - \frac{\nu}{N}\Big)\Big).
\end{multline}
\end{theorem}

\begin{proof}
Let $\varphi_{q,\kappa} (x) = \xi(-qx-\kappa)$ and $t_{p/q} = (2\pi)^{-1} p/q$. As we observed in Remark \ref{sect2-rem2}, one has
$G(-p,\kappa,q) = (\phi(t_{p/q},x)*\varphi_{q,\kappa})(0)$.
It follows from Remark \ref{sect2-rem1} (in particular from
the uniform estimates \eqref{sect2-eqrem1}) that one has also
\begin{equation}\label{sect2-eq12}
G(-p,\kappa,q) = q
\lim\limits_{K\rightarrow +\infty}
\Big(\Big(
\sum\limits_{k=-K}^K e^{-i (2\pi k)^2 t} e^{2ikx}\Big) * \varphi_{q,\kappa}(x)\Big)(0).
\end{equation}
If follows from Proposition \ref{sect1-prop1} that
for any $k\in [-K,K]$, the distribution $[e^{2i\pi kx}]$ evolves
in $\mathscr D_{0,1}'(\R^+_t \times \R_x,\C)$ as
$$
(t,x) \longmapsto \Big[e^{-i (2\pi k)^2 t} e^{2ikx}\Big]=  \big(e^{-2k\pi t D_x}\circ
\mathbb M_{2k\pi} \circ e^{-2k\pi t D_x}\big) ([1]) =
\big(e^{-2k\pi t D_x}\circ
\mathbb M_{2k\pi}\big) ([1]).
$$
For any $N,N'\in \N^*$ and $t\in \R^+$, let for $z,w\in \C$,
\begin{equation}\label{sect2-eq13}
\begin{split}
& M_N(z,k) = \sum\limits_{\nu=0}^N
\binom{N}{\nu} \Big(
\frac{1}{2} + \frac{k}{q}\Big)^{N-\nu}
\Big(\frac{1}{2} - \frac{k}{q}\Big)^\nu \exp
\Big(2i\pi \Big(
\frac{1}{2} - \frac{\nu}{N}\Big) qz\Big) \\
& T_{N'}^t(w,k) = \sum\limits_{\nu'=0}^{N'}
\binom{N'}{\nu'}
\Big(\frac{1}{2} -kp\Big)^{N'-\nu'}
\Big(\frac{1}{2} + kp\Big)^{\nu'}\, \exp
\Big( - \Big(\frac{1}{2}- \frac{\nu'}{N'}\Big)
\frac{2\pi q t}{p}\, w\Big).
\end{split}
\end{equation}
It follows from \cite[Lemma 2.4]{cssy18} (see also \cite[\textsection 3]{CSSY}) that for any
$z,w\in \C$, for any $t\in \R^+$,
\begin{equation}
\label{sect2-eq14}
\begin{split}
& |M_N(z,k)- e^{2i\pi k z}| \leq \frac{2}{3} \,
\frac{|(2\pi k/q)^2-1|}{N} (q|z|)^2 \exp\Big(\big(1+\max (|2\pi k/q|,1)\big)q|z|\Big)\\
&|T_{N'}^t(w,k) - e^{-2\pi k t w}| \leq \frac{2}{3}
\, \frac{|(kp)^2 -1|}{N'} \, \Big(\frac{2\pi t}{p} |w|\Big)^2
\exp \Big(\big(1+ \max(|kp|,1)\big)\,\frac{2\pi t}{p} |w|\Big).
\end{split}
\end{equation}
Specify now the value of $t$ as $(1/(2\pi))\times p/q$.
For any $N,N'\in \N^*$, for any $\nu \in \{0,...,N-1\}$, $\nu'\in \{0,...,N'-1\}$, one has
\begin{equation*}
\Big(\exp \Big(-
\Big(\frac{1}{2} - \frac{\nu'}{N'}\Big)\, \frac{D_x}{q}\Big)
\circ \exp
\Big( 2i\pi\, \Big(
\frac{1}{2} - \frac{\nu}{N}\Big)\, q x\Big)
([1]) = \omega_{\nu,\nu'}^{N,N'} (0)\,
 \Big[\exp
\Big( 2i\pi\, \Big(
\frac{1}{2} - \frac{\nu}{N}\Big)\, q x\Big)\Big].
\end{equation*}
We have then
\begin{multline}\label{sect2-eq15}
\Big(\Big(\exp \Big(-
\Big(\frac{1}{2} - \frac{\nu'}{N'}\Big)\, \frac{D_x}{q}\Big)
\circ \exp
\Big( 2i\pi\, \Big(
\frac{1}{2} - \frac{\nu}{N}\Big)\, q x\Big)
([1]) *\xi_{q,\kappa} (x)\Big)(0) \\
= \omega_{\nu,\nu'}^{N,N'} (0)\,
\Big( \exp
\Big( 2i\pi\, \Big(
\frac{1}{2} - \frac{\nu}{N}\Big)\, q x\Big) *
\xi_{q,\kappa}(x)\Big) (0) \\
= \omega_{\nu,\nu'}^{N,N'}(0)\,
\hat \xi_{q,\kappa}\Big(
2\pi \Big(\frac{1}{2} - \frac{\nu}{N}\Big) q\Big)  =
\frac{\omega_{\nu,\nu'}^{N,N'}(\kappa)}{q}\, \widehat \varphi\Big(
2\pi \Big(\frac{1}{2} - \frac{\nu}{N}\Big)\Big).
\end{multline}
The first inequality in \eqref{sect2-eq14} applied when $k\in [-K,K]$ shows that the sequence of entire functions $(z\in \C\mapsto M_N(z,k))_{N\geq 1}$ converges towards $z\mapsto
e^{2i\pi kz}$ in $A_1(\C_z)$. The family $\{z\mapsto M_N(z,k)\}_{N\geq 1}$
is then a bounded family in $A_1(\C_z)={\rm Exp}(\C_z)$ ; more precisely
one has $|M_N(z,k)| \leq a_{p,q} K^2 \exp (b_{p,q} K\, |z|)$ for any
$z\in \C$, $k\in [-K,K]$ and $N\in \N^*$ for some positive constants
$a_{p,q}$, $b_{p,q}$. Therefore, each
such function $z\mapsto M_N(z,k)$ can be interpreted as the Fourier-Borel transform of some analytic functional $Y_{N}^k$, that is $M_N(z,k) = \langle Y_{N}^k(w), e^{wz}\rangle$ for any $z\in \C$.
Moreover, a carrier for such a functional can be chosen as $[-b_{p,q} K-\varepsilon, b_{p,q} K + \varepsilon]$ for $\varepsilon >0$ arbitrary small.
Then, for any $k\in [-K,K]$,
$$
e^{-2 \pi k t_{p,q} D_z}(M_N(z,k)) = e^{- k p/q D_z}(M_N(z,k)) = \langle e^{-kp w/q}\, Y_N^k(w),e^{wz}\rangle.
$$
The second inequality in \eqref{sect2-eq14}, applied with $t= t_{p,q}(2\pi)^{-1} p/q$, shows that the sequence of entire functions $(w\mapsto T_{N'}^{t_{p,q}}(w,k))_{N'\geq 1}$ converges towards
$w\mapsto e^{-kp w/q}$ in $A_1(\C_w)$, in particular uniformly
on any compact subset of $\C_w$. The family $\{w\mapsto T_{N'}^{t_{p,q}}(w,k)\}_{N'\geq 0}$ is then a bounded family in $A_1(\C_w)$; more precisely, one has $|T^{t_{p,q}}_{N'} (w,k)|\leq a'_{p,q} K^2 \exp (b'_{p,q} K |w|)$ for any $w\in \C$, $k\in [-K,K]$ and $N\in \N^*$ for
some positive constants $a'_{p,q}$, $b'_{p,q}$. For any $z\in \C$,
the sequence of functions $(w\mapsto T^{t_{p,q}}_{N'}(w,k) e^{wz})_{N'\geq 0}$ converges towards $w\mapsto e^{-k pw/q} e^{wz}$ in $H(\C_w)$, uniformly when $z$ is in a compact of $\C_z$.
Moreover the second inequality in \eqref{sect2-eq14} provides an explicit estimate of the error when $k\in [-K,K]$ and $z$ belongs
to a compact subset of $\C$. Therefore the bi-indexed sequence of functions
$$
z \longmapsto \big\langle T^{t_{p,q}}_{N'}(w)
\, Y_N^k(w), e^{wz}\big\rangle\quad (N,N'\geq 1)
$$
converges uniformly on any compact of $z$ towards the function
$z\mapsto (e^{-k p D_z/q}\circ \mathbb M_{2k\pi})(1)(z)$ when $N$ and $N'$ tend to infinity with error estimates provided by the inequalities
\eqref{sect2-eq14} when $k\in [-K,K]$. As a consequence of these considerations,
\begin{multline}\label{sect2-eq16}
\lim_{K\rightarrow +\infty}
\Bigg(\sum\limits_{k=-K}^K
\Bigg(\big(e^{- k p D_x/q}\circ
\mathbb M_{2k\pi}\big) -
\sum\limits_{\nu=0}^{N_K}
\sum\limits_{\nu'=0}^{N_K}
\Big(\frac{1}{2}+ \frac{k}{q}\Big)^{N-\nu}
\Big(\frac{1}{2}- \frac{k}{q}\Big)^{\nu}
\Big(\frac{1}{2} -kp\Big)^{N'-\nu'}
\Big(\frac{1}{2} + kp\Big)^{\nu'} \\
\Big(\exp \Big(-
\Big(\frac{1}{2} - \frac{\nu'}{N'}\Big)\, \frac{D_x}{q}\Big)
\circ \exp
\Big( 2i\pi\, \Big(
\frac{1}{2} - \frac{\nu}{N}\Big)\, q x\Big)\Big)\Bigg)([1])*\varphi_{q,k}(x)\Bigg) =0.
\end{multline}
If one substitutes
\eqref{sect2-eq15} in the evaluation of the right-hand side of
\eqref{sect2-eq16} at $x=0$, one gets the approximation formula \eqref{sect2-eq11} from \eqref{sect2-eq12}.
\end{proof}

\section{Time dependent Schr\"odinger equation with periodic potential}\label{section4}

Let $V~: \R \rightarrow \R$ be a smooth $2\pi$-periodic function, with Fourier development
\begin{equation}\label{sect4-eq1}
V(x) = \sum\limits_{k\in \Z} c_k e^{ikx},
\end{equation}
where the sequence of Fourier coefficients $(c_k)_{k\in \Z}$ is rapidly decreasing on $\R$.
\vskip 2mm
\noindent
Given $\omega \in \R$, the effect of replacing $V$ by $V^{-\omega}$
defined by $(t,x) \mapsto V^{-\omega}(t,x) =V(t,x+2\omega t)$ is to modulate with respect to the variable $t$ with the frequency
$\omega$ the discrete spectrum of $V$, since
$$
V^{-\omega} (t,x) = \sum\limits_{k\in \Z} c_k e^{2ik\omega t}\,
e^{ikx}\quad ((t,x)\in \R^+ \times \R).
$$

\begin{lemma}\label{sect4-lem1} Let $\alpha=\pm 1$,
$T>0$ and $\omega\in \R$. A distribution $\phi\in \mathscr S_x'([0,T[_t\times \R_x,\C)$ which is $2\pi$-periodic in the second variable $x$ belongs to
$\mathscr D'_{V^{-\omega},\alpha}(\R^+_t \times \R_x)$ if and only its
Fourier coefficients $\hat c_k(\phi,t)\in \mathscr D'([0,T[_t,\C)$
with respect to $x$ are such that the distributions
$\psi_k(t,x) = e^{i\alpha k^2 t} \hat c_k(\phi,t)$ satisfy
the homogeneous linear system of equations
\begin{equation}\label{sect4-eq2}
i \frac{\partial \psi_k}{\partial t}
 - \sum\limits_{\ell \in \Z}
c_\ell\, e^{2i\ell \omega t}
e^{i\alpha \ell (2k-\ell) t}\, \psi_{k-\ell} = 0\quad (k\in \Z).
\end{equation}
\end{lemma}

\begin{proof}
Since the distribution $(t,x)\mapsto \phi(t,x)$ is periodic in $x$,
one can consider it as an element of $\mathscr D'([0,T[_t\times \R/(2\pi \Z)$ and thus define its Fourier coefficients $\hat c_k(\phi,t)\in \mathscr D'([0,T[_t,\C)$ for $k\in \Z)$ by
$$
\langle c_k(\phi,t)\,,\, \theta(t)\rangle =
\langle \phi(t,x)\,,\, \theta(t)\otimes e^{-ikx}\rangle
$$
for all test-function $\theta \in \mathscr D(]-\infty,T[_t,\C)$. For any
such test function $\theta$, the fact that $\phi$ belongs to $\mathscr S'_{V^{-\omega},\alpha}(\R_t^+ \times \R_x,\C)$ implies that, for any $k\in \Z$,
\begin{eqnarray*}
i \Big\langle \frac{d \hat c_k(\varphi,\cdot)}{dt}(t),\theta (t)\Big\rangle
&=& - i\langle \hat c_k(\phi,t),\theta'(t)\rangle
=  i \Big\langle \Big(\frac{\partial}{\partial t}\Big)(\phi)(t,x)\,,\,
\theta(t)\otimes e^{ikx}\Big\rangle \\
&=&  - \alpha \Big\langle
\Big(\frac{\partial^2}{\partial x^2}\Big)(\phi) (t,x)\,,\,
\theta(t)\otimes e^{-ikx}\Big\rangle +
\Big\langle \phi(t,x)\,,\, \theta(t)\otimes e^{-ikx}\, V^{-\omega}(t,x))\Big\rangle \\
&=& \alpha\, k^2 \big\langle \hat c_k(\phi,t)\,,\, \theta(t)\big\rangle
+ \sum\limits_{\ell \in \Z} c_\ell\,
\big\langle e^{2i\omega \ell t}\, \hat c_{k-\ell}(\phi,t)\,,\, \theta(t)\big\rangle.
\end{eqnarray*}
Such conditions are in fact necessary and sufficient.
The lemma follows from the definition of the distributions
$\psi_k$.
\end{proof}

\noindent
For any $\ell\in \Z$, let $A_{\alpha,\ell}(t)$ be the infinite matrix
$A_{\alpha,\ell} = (a^{\alpha,\ell}_{\kappa_1,\kappa_2})_{\kappa_1\in \Z,\kappa_2\in \Z}$
which is defined as
$$
A_{\alpha,\ell}(t) = Z^\ell \cdot {\rm Diag} [c_\ell\, e^{i\alpha \ell(2k-\ell)t}\,;\, k\in \Z],
$$
where $Z$ is the matrix of the shift $(\gamma_k)_{k\in \Z}
\longmapsto (\gamma_{k-1})_{k\in \Z}$. If one denotes as
as $t\mapsto \Psi(t)$ the (infinite) column matrix-valued function with entries the $t\mapsto \psi_k(t)$, then the (infinite) differential system
\eqref{sect4-eq2} can be  formulated as
\begin{equation}\label{sect4-eq4}
\Psi'(t) = -i \Big(\sum\limits_{\ell \in \Z}
e^{2i\omega \ell t}\, A_{\alpha,\ell}(t)\Big)\cdot \Psi(t) = -i A_{\alpha}^\omega(t)\cdot \Psi(t),
\end{equation}
where $A_{\alpha}^\omega(t) = \sum_{\ell\in \Z} e^{2i\omega \ell t} A_{\alpha,\ell}(t)$.
\vskip 2mm
\noindent
If $|A_\alpha^0|$ denotes the matrix indexed by $\Z^2$ which entries
are the absolute values of the entries of $A_\alpha^0$, then one has for any $t\in \R^+$ that
\begin{equation}\label{sect4-eq4bis}
\begin{split}
&\|\, |A_\alpha^0(t)|\,\|_\infty = \|(c_k)_{k\in \Z}\|_\infty \\
&\|(|A_\alpha^0(t)|)^n| \leq \|\, |A_\alpha^0(t)|\, \|_\infty \|(c_k)_{k\in \Z}\|^{n-1}_1\quad \forall\, n\geq 1.
\end{split}
\end{equation}
For any $M\in \N^*$, let
\begin{equation}\label{sect4-defB}
B_{\alpha}^{\omega,M}(t) = \int_0^t
\Big( \sum\limits_{\ell =-M}^M e^{2i\omega\ell \tau}
A_{\alpha,\ell}(\tau)\Big)\, d\tau
\end{equation}
(understood here as the Bochner integral of a $\ell^\infty_\C(\Z^2)$-valued function). For any $M, n\in \N^*$, consider
\begin{multline}\label{sect4-eq4ter}
\stackrel{n\ {\rm times}}{\overbrace{B_\alpha^{\omega,M} (t) \times \cdots \times
B_\alpha^{\omega,M} (t)}} =
\int_{[0,t]^n} \sum\limits_{\ell_1=-M}^M \cdots
\sum \limits_{\ell_n=-M}^M
e^{2i\omega \langle \ell,\tau\rangle} A_{\alpha,\ell_1}(\tau_1)\times \cdots
\times A_{\alpha,l_n}(\tau_n)\, d\tau
\end{multline}
One has
$$
\sum\limits_{\ell_1=-M}^M \cdots
\sum \limits_{\ell_n=-M}^M
|A_{\alpha,\ell_1}(\tau_1)|\times \cdots \times
|A_{\alpha,l_n}(\tau_n)|
\leq (\|(c_k)_{k\in \Z}\|_1)^n \quad \forall\, M\in \N,\ \forall \, n\in \N^*.
$$
Then, for any $t\in \R^+$, one has that
$$
\lim\limits_{M,N\rightarrow +\infty} \Big({\rm Id} +
\sum\limits_{n=1}^N
\frac{\stackrel{n\ {\rm times}}{\overbrace{B_\alpha^{\omega,M} (t) \times \cdots \times
B_\alpha^{\omega,M} (t)}}}{n!}\Big)  = \exp \Big(
\int_0^t A_{\alpha}^\omega (\tau)\, d\tau\Big),
$$
where the convergence is here a normal convergence in the Banach space
$\ell_\C^\infty(\Z^2)$, with
\begin{equation}\label{sect4-eq4q}
1 + \sum\limits_{n=1}^N
\frac{\stackrel{n\ {\rm times}}{\overbrace{|B_\alpha^{\omega,M} (t)| \times \cdots \times
|B_\alpha^{\omega,M} (t)|}}}{n!}
\leq \exp (\|(c_k)_{k\in \Z}\|_1 t).
\end{equation}
One can now formulate the following lemma.

\begin{lemma}\label{sect4-lem2} The distribution
\begin{equation}\label{sect4-eq5}
(t,x)\in \R^+ \times \R \longmapsto \Phi^{\alpha,\omega} (t,x) := \sum\limits_{k\in \Z}
e^{-i\alpha k^2 t}
\Big(\Big(\exp \Big(
\int_0^t A_\alpha^\omega (\tau)\, d\tau\Big)\Big) (\delta_0(\kappa))_{\kappa\in\Z}\Big)_k\, e^{ikx}
\end{equation}
which is in $\mathscr S_x'(\R^+_t\times \R_x,\C)$ belongs
to $\mathscr D'_{V^{-\omega},\alpha}(\R^+_t\times \R_x,\C)$ and realizes an evolution of the distribution $[1]\in \mathscr D'(\R)$ from $t=0$. Such evolution is the unique one which remains $2\pi$-periodic in $x$.
\end{lemma}

\begin{proof}
As we have observed above, the function
$$
t \in \R^+ \longmapsto \exp \Big(
\int_0^t A_\alpha^\omega (\tau)\, d\tau\Big) \in \ell^\infty_\C(\Z^2)
$$
is well defined and dominated by $t\mapsto \exp (\|(c_k)_{k\in \Z}\|_1\, t)$ on $\R^+$. The coordinates $(t\mapsto \psi_k(t))_{k\in \Z}$ of the  image of $(\delta_0(\kappa))_{\kappa \in \Z}$
satisfy the differential system \eqref{sect4-eq1} and are such that $\psi_k(0) = 1$ if $k=0$, $0$ otherwise, which concludes the proof of the lemma. Unicity follows from the construction itself.
\end{proof}

\begin{remark}\label{sect4-rem1}
{\rm The conclusion of Lemma \ref{sect4-lem2} subsists when $V$ is just supposed to be in $\mathscr C^2(\R,\C)$, in which case its spectrum
$(c_k)_{k\in \Z}$ belongs to $\ell^1(\Z)$, which is sufficient to ensure the validity of the lemma.
}
\end{remark}

\noindent
We are now in situation to profit from the twisting concept which has been introduced in \textsection \ref{section2}.

\begin{theorem}\label{sect4-thm2} Let $V$ be a real potential on $\R$ of class $\mathscr C^\infty$. The distribution $[e^{i\omega x}]\in \mathscr S'(\R_x,\C)$ evolves in a unique way as an element in $\mathscr S'_x(\R^+_t\times \R_x,\C)
\cap \mathscr D'_{V,\alpha}(\R^+_t\times \R_x)$ which is given by
\begin{equation}\label{sect4-eq5bis}
\phi^{\alpha,\omega} (t,x)
=  \sum\limits_{k\in \Z}
e^{-i\alpha k^2 t}
\Big(\Big(\exp \Big(
\int_0^t A_\alpha^\omega (\tau)\, d\tau\Big)\Big) (\delta_0(\kappa))_{\kappa\in\Z}\Big)_k  e^{-i\omega^2 t} e^{-2i\omega t}
e^{i(k+\omega)x}.
\end{equation}
\end{theorem}

\begin{proof}
It follows from Corollary \ref{sect1-cor1} that
$(t,x)\mapsto (\mathbb G^\omega_{V^{-\omega}})(\phi^{\alpha,\omega})(t,x)$ (which leads precisely to the expression
\eqref{sect4-eq5bis})
is a possible evolution for $[e^{i\omega x}]$ in
$\mathscr D'_{V,\alpha}(\R^+_t\times \R_x,\C)$. Since the potential $V$ is real, hence the corresponding multiplication operator is self-adjoint, such an evolution is unique, which concludes the proof of the lemma.
\end{proof}

\begin{remark}\label{sect4-rem2}
{\rm The conclusion of Theorem \ref{sect4-thm2} subsists when $V$ is just supposed to be in $\mathscr C^2(\R,\C)$, in which case its spectrum
$(c_k)_{k\in \Z}$ belongs to $\ell^1(\Z)$, which is sufficient to ensure the validity of Lemma \ref{sect4-lem2}, hence of Theorem \ref{sect4-thm2}.
}
\end{remark}

\noindent
Thanks to the closed formula for the evolution of $[e^{i\omega x}]$ which is provided by Theorem \ref{sect4-thm2}, we may also state a result
in view of the supershift context. Let us define, for any
$K,M,N\in \N^*$ the following truncated version of $\phi^{\alpha,\omega}$.
\begin{multline}\label{sect4-eq6}
\phi_{K,M,N}^{\alpha,\omega}
= \sum\limits_{k=-K}^K
e^{-i\alpha k^2 t}
 \Big(\Big({\rm Id} +
\sum\limits_{n=1}^N
\frac{\stackrel{n\ {\rm times}}{\overbrace{B_\alpha^{\omega,M} (t) \times \cdots \times
B_\alpha^{\omega,M} (t)}}}{n!}\Big)(\delta_0(\kappa))_{\kappa\in \Z}\Big)_k\, e^{-i\omega^2 t} e^{-2i\omega t}
e^{i(k+\omega)x},
\end{multline}
where the functions $t\in \R^+ \mapsto B_\alpha^{\omega,M}(t)$ have been introduced in \eqref{sect4-defB}. We can then state the following result.

\begin{theorem}\label{sect4-thm3} Let $\alpha =\pm 1$ and
$V$ be a real smooth potential $x\mapsto V(x)$. For any
$\omega \in \R$, let $\phi^{\alpha,\omega}$ be the evolution of $[e^{i\omega x}]$ in $\mathscr S'_x(\R^+_t\times \R_x,\C)
\cap \mathscr D'_{V,\alpha}(\R^+_t\times \R_x)$ as explicited
in \eqref{sect4-eq5bis} and truncated in \eqref{sect4-eq6}. Then, the two following assertions hold for any $\omega \in \R$.
\begin{itemize}
\item  one has
$$
\lim\limits_{K,M,N \rightarrow \infty\ ({\rm each\ of\ them})} \phi^{\alpha,\omega}_{K,M,N} = \phi^{\alpha,\omega}
$$
in $\mathscr S'_x(\R^+_t\times \R_x,\C)$, that is when considered as distributions on $\R^+_t\times \mathbb S^1$ ;
\item  one has also
$$
\lim\limits_{N'\rightarrow +\infty}
\sum\limits_{\nu=0}^{N'} \binom{N'}{\nu}
\Big(\frac{1 + \omega}{2}\Big)^{N'-\nu}
 \Big(\frac{1 - \omega}{2}\Big)^{\nu} \phi^{\alpha, 1-2\nu/N'}_{K,M,N} =
\phi_{K,M,N}^{\alpha,\omega}
$$
in $\mathscr D'(\R^+_t\times \R_x,\C)$ for any $K,M,N$ in $\N^*$.
\end{itemize}
\end{theorem}

\begin{proof} The first assertion follows from the estimates
\eqref{sect4-eq4q}. The second assertion follows from the fact that
for each $k\in [-K,K]$, an expression such as
$$
 \Big(\Big({\rm Id} +
\sum\limits_{n=1}^N
\frac{\stackrel{n\ {\rm times}}{\overbrace{B_\alpha^{\omega,M} (t) \times \cdots \times
B_\alpha^{\omega,M} (t)}}}{n!}\Big)(\delta_0(\kappa))_{\kappa\in \Z}\Big)_k\, e^{-i\omega^2 t} e^{-2i\omega t}
e^{i\omega x}
$$
can be interpreted as the action on $e^{i\omega x}$ of a differential operator (with coefficients depending on $t$ but uniformly controlled when
$t$ belongs to a compact subset of $\R^+$) which symbol lies in $A_2(\C)$. This follows from the explicit expression
for
$$
\frac{\stackrel{n\ {\rm times}}{\overbrace{B_\alpha^{\omega,M} (t) \times \cdots \times
B_\alpha^{\omega,M} (t)}}}{n!}
$$
given in \eqref{sect4-eq4ter}. Such an operator acts continuously from
$A_1(\C)$ into itself (see \cite{AOKI,QS2,cssy18}) and then propagates the fact that
the sequence of entire functions
$$
z \longmapsto \sum\limits_{\nu=0}^{N'}  \binom{N'}{\nu}
\Big(\frac{1 + \omega}{2}\Big)^{N'-\nu}
 \Big(\frac{1 - \omega}{2}\Big)^{\nu} e^{i(1-2\tau/N')z}
$$
converges towards $z\mapsto e^{i\omega z}$ in $A_1(\C)$ when
$N'$ tends to $+\infty$.
\end{proof}

\begin{remark}\label{sect4-rem3}
{\rm It is impossible because of trivial spectral considerations already mentioned that, given $\omega \in \R$, the sequence
$$
\sum\limits_{\nu=0}^{N'} \binom{N'}{\nu}
\Big(\frac{1 + \omega}{2}\Big)^{N'-\nu}
 \Big(\frac{1 - \omega}{2}\Big)^{\nu} \phi^{\alpha, 1-2\nu/N'}
$$
converges towards $\phi^{\alpha,\omega}$ in
$\mathscr S'_x(\R^+_t\times \R_x)$.
As for the convergence of this sequence towards $\phi^{\alpha,\omega}$ in
$\mathscr D'(\R_t\times \R_x,\C)$, it seems unlikely that it could be true. The symbol of the operator that should be involved (instead of a differential operator with symbol in $A_2(\C)$ as we use here) should be in $A_{\chi}$, where $\chi(z)= \exp (|z|)$, since a double exponentiation occurs in the construction. One may probably be able to formulate a convergence within the frame of hyperfunctions or, more probably, ultradistributions.
The question remains open, and we will return to it in a future paper.
}
\end{remark}

\Addresses

\end{document}